\documentclass[11pt]{article}
\usepackage{hyperref}
\usepackage{amsfonts,amsmath,amsthm,wrapfig}
\usepackage{comment}
\usepackage{fullpage}
\usepackage[noend]{algorithmic}
\usepackage[boxed]{algorithm}
\usepackage{xcolor}

\usepackage{enumerate}

\newtheorem{theorem}{Theorem}
\newtheorem{lemma}[theorem]{Lemma}

\newcommand\E{\mathbb{E}}
\newcommand\R{\mathbb{R}}
\newcommand\cs{\mathcal{S}}

\newcommand\disc{\mathrm{disc}}

\newcounter{note}[section]

\begin{document}

\title{On the discrepancy of random low degree set systems}
\author{Nikhil Bansal\thanks{CWI and TU Eindhoven, Netherlands.
\texttt{bansal@gmail.com}. Supported by a NWO Vidi grant 639.022.211 and an ERC consolidator grant 617951.}
 \and Raghu Meka \thanks{UCLA. \texttt{raghum@cs.ucla.edu}. Supported by NSF grant CCF-1553605.}
}
\maketitle
\begin{abstract}
Motivated by the celebrated Beck-Fiala conjecture, we consider the random setting where there are $n$ elements and $m$ sets and each element lies in $t$ randomly chosen sets. In this setting, Ezra and Lovett showed an $O((t \log t)^{1/2})$ discrepancy
bound in the regime when $n \leq m$ and an $O(1)$ bound when $n \gg m^t$.

In this paper, we give a tight $O(\sqrt{t})$ bound for the entire range of $n$ and $m$, under a mild assumption that $t = \Omega (\log \log m)^2$. The result is based on two steps. First, applying the partial coloring method to the case when $n = m \log^{O(1)} m$ and using the properties of the random set system we show that the overall discrepancy incurred is at most $O(\sqrt{t})$. Second, we reduce the general case to that of $n \leq m \log^{O(1)}m$ using LP duality and a careful counting argument. \end{abstract}


\section{Introduction}
Let $(V,\cs)$ be a set system with $V = [n]$ and $\cs=\{S_1,\ldots,S_m\}$ a collection of subsets of $V$. For a two-coloring $\chi: V \rightarrow \{-1,+1\}$, the discrepancy of  a set $S$ is defined as $\chi(S) = |\sum_{i\in S} \chi(i)|$, and measures the imbalance from an even-split of $S$.
The discrepancy  of the system $(V,\cs)$ is defined as
\[\disc(\cs) = \min_{\chi: V \rightarrow \{-1,+1\}} \max_{ S \in \cs} \chi(S).\]
That is, it is the minimum imbalance  of all sets in $\cs$ over all possible two colorings $\chi$.

Discrepancy is a widely studied topic and has applications
to many areas in mathematics and computer
science. For more background we refer the reader to
the books \cite{Chazelle,MAT99,Panorama}. In particular, discrepancy is closely
related to the problem of rounding fractional solutions
to a linear system of equations \cite{LSV},
and has found several applications in approximation algorithms and
optimization.

An important problem, motivated by the rounding fractional solutions to column-sparse linear systems, is
to understand the discrepancy of sparse systems where each element $i\in [n]$ lies in at most $t$ sets.
In a classic result, Beck and Fiala \cite{BF81} showed that the discrepancy of such systems is at most $2t-1$.
This bound was recently improved by Bukh  to $2t-\log^*t$ \cite{Bukh16}.
Improved bounds with a better dependence on $t$, but at the expense of dependence on $n$, are also known and after long line of work the best such bound is $O(t^{1/2} (\log n)^{1/2})$ due to Banaszczyk \cite{B97}.
These results have also been made algorithmic in recent years \cite{B10,BDG16,LevyRR17}.

It is a long-standing conjecture that the discrepancy of such set systems
is $O(t^{1/2})$ \cite{BF81}.
Despite much work, the problem is open even for very special cases such as when the hypergraph corresponding to the set
system is simple, i.e.~any two sets intersect in at most one element. Another interesting question to get the tight $O(t^{1/2})$ bound in the case when we have
the additional property that the sets also size at most $t$. Here the best known bound is $O((t \log t)^{1/2})$ based on a direct application of the Lov\'{a}sz Local Lemma.

\paragraph{Random set system model.}
Recently, Ezra and Lovett \cite{EL15} consider the problem in a natural random model, where there are $n$ elements and $m$ sets and each element $i \in [n]$ lies in exactly $t$ random sets. That is, the $t$-tuple of sets containing $i$ is chosen uniformly at random among the $\binom{m}{t}$ possibilities. In the following, by a random set system we refer to this model.

Ezra and Lovett \cite{EL15} proved the following two results in the random model.
(i) For $n\leq m$, the expected discrepancy is $O (({t \log t})^{1/2})$, and (ii) for $n \gg m^t$, the expected discrepancy is $O(1)$. We remark that an $\Omega(t^{1/2})$ lower bound on the expected discrepancy also holds in the random model (e.g.~when $n=m=2t$, as can be seen easily using the spectral lower bound method \cite{Chazelle}).

There are two natural questions left open from their work.
First, whether these results can be extended to the entire range of $n$ and $m$, i.e.~when $n \in [m,m^t]$.
This is particularly interesting, as the result of \cite{EL15} in the regime when $n \leq m$ is based on Lov\'{a}sz Local Lemma, which fails
for inherent reasons\footnote{As the average set size is $nt/m \gg t$.} when $n \gg m$.
A second natural question is whether their bound can be improved to the optimum bound of $O(t^{1/2})$, especially for the important case of $n = m$.
Again, the local lemma inherently loses an additional $(\log t)^{1/2}$ factor when $n=\Theta(m)$.

\subsection{Our results and overview}
Our main result addresses both these questions, and is the following.

\begin{theorem}
\label{thm:main} Let $(V,\cs)$ be a random set system on $n$ elements and $m$ sets, where each element lies in $t$ sets. Then, for every $n$ and $m$, the expected discrepancy of the set system is  $O(t^{1/2})$,  provided that $t = \Omega((\log \log m)^2)$.
\end{theorem}

In particular, this gives the tight $O(t^{1/2})$ bound for the entire range of $n$ and $m$, assuming $t \geq  \Omega((\log \log m)^2)$. Moreover, the algorithm can be implemented in randomized polynomial time.
The result is based on two main ideas.


%

\paragraph{Reduction of $n$ to $k$.} We show that the problem with arbitrary $n$, $m, t$ can be {\em reduced} to the case of $n \leq k$, where $k= C m  \log^2 m$, with high probability, for a fixed constant $C$.
More precisely, let $A$ be the $m \times n$ incidence matrix of the random set system, and let $a_i$ denote the $i$-th column of $A$. We start by applying the Beck-Fiala theorem \cite{BF81} to the elements $\{k+1,\ldots,n\}$
to find a $\{-1,1\}$ coloring $\chi'$ with discrepancy at most  $2t-1$,~i.e.$\|\sum_{i>k} \chi'(i) a_i \|_\infty  \leq 2t-1$.
Let us denote this discrepancy vector by $b \in [-2t+1,2t-1]^m$.

We show that with high probability,
there is a {\em fractional} coloring $\chi''$ (i.e.~with colors in $[-1,1]$) of the elements $\{1,\ldots,k\}$ with discrepancy exactly $-b$.
Together this gives a coloring $\chi$ with discrepancy $0$, where the elements $k+1,\ldots,n$ are colored $\pm 1$, but the elements $1,\ldots,k$ have colors in $[-1,1]$. As the first $k$ columns are still random, this gives a ``reduction" of the random Beck-Fiala problem for general $n$ to that for $k$.

The existence of the coloring $\chi''$ follows from the following result of independent interest.
\begin{theorem}
\label{thm:1}
For all $c > 0$, there exists a constant $C > 0$ and $c' \in (0,1)$ such that the following holds.
Let $a_1,\ldots,a_k \in \{0,1\}^m$ be random vectors with $t$ ones. Let $P := \{ \sum_i a_i x_i :   \  x_i \in [-1,1]\} \in \R^m$ be the set of discrepancy vectors achievable by fractional colorings of $a_1,\ldots,a_k$. Then for $k \geq C m \log^2 m$, with probability at least $1 - 1/m^c$, it holds that  $ 2t B_\infty^m \subset P$, where $B_\infty^m = \{y \in \R^m:\|y\|_\infty \leq 1\}$ is the $\ell_\infty$ ball in $\R^m$.
%
%
\end{theorem}


To prove Theorem \ref{thm:1}, we use LP duality to give an equivalent condition for the property $ 2t B_\infty^m \subset P$.
Next, we use a counting argument to show that this condition is satisfied with high probability for $k$ random vectors. We first prove a weaker bound of $k = O(m^4\log m)$ using a standard $\epsilon$-net argument. Later, we give a much more careful argument to improve this bound to to $k = O(m \log^2 m)$.

\paragraph{Partial Coloring.} It remains to modify the fractional coloring $\chi''$ on $[k]$ to an integral $\{-1,+1\}$ coloring, while incurring low discrepancy.
To achieve this, we apply the partial coloring procedure of Lovett and Meka \cite{LM15} over $O(\log k)$ iterations.
The main issue here is  to ensure that the overall discrepancy stays bounded by $O(t^{1/2})$ over all the iterations.
To this end, we use the property that the starting set system on $k$ columns is random to control the potential used in the partial coloring lemma of Lovett and Meka.
However, as the partial coloring method gives no control on which subset of the original $k$ columns remain after each iteration, we incur a penalty due to a union bound over all subsets of the original columns. This is where we require that  $k$ is not too large relative to $m$.
In particular, we show the following.
\begin{theorem}
\label{thm:2}
Let $A$ be a random $m\times k$ matrix where each column has $t$ ones. Then, if $t = \Omega(\log^2 (ek/m))$ then with probability at least $1-\exp(-t)$, the discrepancy of $A$ is $O(t^{1/2})$.
\end{theorem}

Combining Theorems \ref{thm:1} and \ref{thm:2} directly gives Theorem \ref{thm:main}.
In particular, the condition $t = \Omega((\log \log m)^2)$ arises as $t = \Omega(\log^2 (ek/m))$ and $k/m = \log^{C}m$.

\paragraph{Limitations.}
By a more elaborate algorithm and case analysis, the
lower bound  on $t$  in Theorem \ref{thm:2} can be improved to $t \geq \tilde{\Omega}(\log \log m)$, where $\tilde{\Omega}(\cdot)$ hides lower order factors. However, we do not describe these more complex calculations here as $ t \approx \log \log m$ is a natural bottleneck for our methods. 
In particular, $k = \Omega( m \log m)$ is necessary for Theorem \ref{thm:1} to hold, and in this setting, if $t \ll \log \log m$, then there exists several $\ell \times \ell $ submatrices of $A$, where the average row size is $\Omega(t)$,
leading to a discrepancy of $t^{1/2}$  that could add up over the iterations.

\subsection{Preliminaries}
We will need the following probabilistic tail bound.
\begin{lemma}(Bernstein's inequality.)
\label{lem:bernstein}
If $X_1,X_2,\ldots,X_n$ are independent real-valued random variables with $|X_i - \E[X_i]| \leq M$, $\sigma_i^2 = \E[X_i^2] - \E[X_i]^2$. Then
\[ \Pr[\sum_i (X_i - \E[X_i]) > t] \leq \exp\left(\frac{-t^2/2}{( \sum_i \sigma_i^2 + Mt/3)} \right) \]
\end{lemma}
The lower tail follows by replacing $X$ by $-X$ above.

\paragraph{Stochastic Dominance.} For non-negative random variables $X$ and $Y$, we say that $X$ stochastically dominates $Y$ if for all $a>0$,
$\Pr[X > a] \geq \Pr[Y>a]$. We will use this as follows.

Let $X_1,\ldots,X_n$ be independent copies of $X$ and let $Y_1,\ldots,Y_n$ be independent copies of $Y$. Then for any $t>0$,
\[ \Pr[X_1 + \ldots + X_n \geq t] \geq \Pr[Y_1 + \ldots + Y_n \geq t]. \]

%
%
\paragraph{Partial Coloring Lemma.}
The algorithmic partial coloring lemma due to Lovett and Meka \cite{LM15}, takes as input some fractional coloring and
target discrepancy bounds for each row, and finds another partial coloring satisfying these row-wise discrepancy bounds and
where at least half the variables are set to $-1$ or $+1$.

\begin{lemma}(Partial Coloring Lemma \cite{LM15}).
\label{lem:pcol}
 Let $v_1,\ldots,v_m \in \R^n$, and $x_0 \in [-1,1]^n$ be a starting point. Let $c_1,\ldots,c_m$ be parameters such that $\sum_{j=1}^m \exp(-c_j^2/16) \leq n/16$, and let $\delta >0$. Then there exists an efficient randomized algorithm that runs in time $O((m+n)^3 \cdot \delta^{-2} \cdot \log (nm/\delta))$ and with probability at least $0.1$ finds a point $x\in [-1,1]^n$ such that
\begin{enumerate}
\item $|\langle x-x_0,v_j \rangle | \leq c_j \|v_j\|_2$ for each $j\in [m]$.
\item $|x_i| \geq 1-\delta$ for at least $n/2$ indices $i\in [n]$.
\end{enumerate}
\end{lemma}

Note that the probability of success can be boosted by running the algorithm multiple times,  and we will assume that the
probability of failure of the algorithm is exponentially small.
When $|x_i| \geq 1-\delta$, we say that variable $i$ is frozen and otherwise it is alive.
Setting $\delta=1/n$, rounding the frozen variables to the nearest $-1$ or $+1$ at the end of the algorithm
can lead to an additional discrepancy of at most  $1$. So we will
assume henceforth that $\delta=0$.

\section*{Related Work}
Very recently, two other groups \cite{RH18,F18} have independently obtained related results. These
results consider the regime where $n \gg m$, and use fourier analytic methods to show
that an $O(1)$ discrepancy can be achieved
for random low degree systems for $n=\Omega(m^2)$ \cite{RH18} and $n=\Omega(m^3)$ \cite{F18}.
Their results are non-algorithmic.


\section{Applying Partial Coloring}
We now prove Theorem \ref{thm:2}. We will in fact show a strengthening of the theorem that gives small discrepancy from any \emph{starting} fractional coloring $x^{(0)}$ as will be required in our reduction from Theorem \ref{thm:1}.
\begin{theorem}
\label{thm:2m}
Let $A$ be a random $m\times k$ matrix where each column has $t$ ones and $x^{(0)} \in [-1,1]^k$. Then, if $t = \Omega(\log^2 (ek/m))$ then with probability at least $1-\exp(-t)$, there exists $\chi \in \{1,-1\}^k$ such that for all rows $v_j$ of $A$, $|\langle v_j, \chi - x^{(0)} \rangle | = O(\sqrt{t})$.
\end{theorem}
\subsection{The Algorithm}
Our input consists of the
random matrix $A^{(0)}$ of at most $n_0 \leq k$ columns, and some fractional coloring $x^{(0)} \in [-1,1]^{n_0}$.

We will apply the partial coloring lemma in several iterations, where at least half the remaining variables become frozen at each iteration. At the beginning of  iteration $i$,
let $n_i$ denote the number of alive variables and let $A^{(i)}$ and $x^{(i)}$ denote the matrix and fractional coloring restricted to those columns.
We use $j$ to index the rows.

The iterations of the algorithm can be divided into three different phases: (i) $i=0$,  (ii) $1 \leq i \leq \log t$ and (iii) $i > \log t$. We now describe each of these phases.

\begin{enumerate}
\item {\em Phase 0.}
Here the input is $A^{(0)}$ and the starting coloring $x^{(0)}$.
We reduce the number of fractional variables to $n_1 \leq m$, by picking any basic feasible solution to the linear program
\[A^{(0)}x=A^{(0)} x^{(0)}  \qquad \text{ subject to } -1 \leq x_i \leq 1 \quad\text{for } i\in [n_0]\]
As $A^{(0)}x=A^{(0)} x^{(0)}$ consists of at most $m$ linearly independent  constraints, the solution will have at most $m$ variables that are not set to $-1$ or $1$.

Note that the resulting matrix $A^{(1)}$ is no longer random. Nevertheless, we will be able to argue that we can still bound the potential required in the partial coloring lemma as will be shown via Lemma \ref{lem:prob}. 

For notational clarity, in subsequent iterations we can assume that $n_i = m 2^{1-i}$  (as $A^{(1)}$ is no longer random we can add columns with all entries $0$ if necessary).

\item {\em Phase 1.}
 For each $i=1,\ldots, \log t$, in iteration $i$ we apply the algorithm in Lemma \ref{lem:pcol} to $A^{(i)}$ and $x^{(i)}$ with the discrepancy bound
\[c_j \|v_j\|_2 = ct^{1/2}/i^2,\]
 where $c$ is some fixed constant that will be specified later.
If the condition $ \sum_{j=1}^m \exp(-c_j^2/16) \leq n_i/16$ in Lemma \ref{lem:pcol} is not satisfied, we
declare fail and abort the algorithm.

If the algorithm does not abort in any iteration, clearly the over all discrepancy during these phases is $ct^{1/2} \sum_i i^{-2} = O(t^{1/2})$.
\item {\em Phase 2.} For $i> \log t$, we apply partial coloring
with $c_j=0$ for sets larger than $ct^{1/2}$
and $c_j = \infty$ otherwise.
Again, the algorithm aborts if $\sum_{j=1}^m \exp(-c_j^2/16) \leq n_i/16$ does not  hold during any iteration.

Assuming the algorithm does not abort, this phase also adds at most $O(t^{1/2})$ discrepancy as a set incurs zero discrepancy as long as its size exceeds $O(t^{1/2})$.
 \end{enumerate}

It is clear by the description of the algorithm that the total discrepancy of any set is $O(t^{1/2})$. So our goal will be to show that
the probability that the algorithm aborts is at most $\exp(-t)$.
If the algorithm aborts, then we simply output the $O(t)$ discrepancy coloring given by the Beck-Fiala Theorem \cite{BF81}. Clearly, the expected discrepancy
of the resulting algorithm is $ O(t \exp(-t)) + O(t^{1/2}) = O(t^{1/2})$.

\subsection{Analysis}
We begin with a simple lemma that we will use repeatedly later.
 \begin{lemma}
\label{lem:prob}
 Let $M$ be some fixed $r \times \ell$ submatrix of an $m \times \ell$ random matrix $A$ where each column has $t$ ones.
 For  $s  \geq 2t \ell/m$,  let $B(s)$ denote the event that each row of $M$ contains at least $s$
 $1$'s. Then over the random choice of $A$,
 \[\Pr[B(s)] \leq \exp( - rs \log (sm/t\ell)/4)   \]
 \end{lemma}
 \begin{proof}
For $i\in [\ell]$, let $X_i$ denote the number of $1$'s in column $i$ of $M$.  Each $X_i$ is independent and has the hypergeometric distribution $H(m,t,r)$ with mean $\E[X_i]=tr/m$.
Using the fact that $H(m,t,r)$ is more sharply concentrated around its mean than the corresponding binomial distribution $\text{Bin}(r,p)$ with $p=t/m$ (\cite{FK15}, page 395), we can bound the upper tail of $\sum_{i=1}^{\ell} X_i$ by the upper tail of $\text{Bin}(r\ell,p)$.

Moreover, as $B(s)$ implies that $\sum_{i=1}^\ell X_i \geq rs$, we  have that $\Pr[B(s)] \leq \Pr[\sum_{i=1}^\ell X_i  \geq rs]$.
By standard Chernoff bounds, with $\mu = pr\ell$ and for any $\delta>0$
\[\Pr[ \text{Bin}(r\ell,p) \geq (1+\delta) \mu ] \leq \exp(-(\mu \delta \log(1+\delta))/2) \]
Setting $(1+\delta) = s/(p\ell) = sm/t\ell$ and as $\delta \geq sm/(2t\ell)$ by our assumption that $sm/t\ell \geq 2$,
\[ \Pr[B(s) \leq \exp( - rs \log (sm/t\ell)/4).\]
\end{proof}

%

Let us first analyze the failure probability in phase $2$.
\begin{lemma} The probability that the algorithm fails during phase 2 is at most $\exp(-t)$.
\end{lemma}
\begin{proof}
Consider some iteration $i$ for $i>\log t$. Let $\ell = n_i = m2^{1-i}$. The iteration $i$ aborts if the number of rows with size $s > ct^{1/2}$ exceeds $\ell/16$. Call such rows big and let $r=\ell/16$.

By Lemma \ref{lem:prob} and a union bound  such $r \times \ell$ submatrices of $A^{(0)}$, this probability is at most
\[ \binom{k}{\ell} \binom{m}{r}  \cdot \Pr[B(s)]
\leq   \left( \frac{ek}{\ell} \right)^{2 \ell}  \cdot \exp(-\ell s/64 \log (sm/t\ell)). \]
Let us define the parameter $\gamma = ek/m$. Writing $ek/\ell = \gamma (m/\ell)$  and assuming $c \geq 256$ (and hence $s/128 \geq 2t^{1/2}$, this is at most
\begin{eqnarray*}
& & \exp (2 \ell   ( \log \gamma + \log m/\ell - s/128 -  \log sm/t\ell))\\
& \leq &  \exp (-2 \ell (2t^{1/2} - \log \gamma -  \log t))
\end{eqnarray*}
By the assumption in Theorem \ref{thm:2} that $\log \gamma \leq t^{1/2}$, this is at most $\exp(-\ell t^{1/2})$.
As $\ell = m 2^{1-i}$ in iteration $i$ and as the phase becomes trivial when $\ell \leq 2t^{1/2}$, the over all probability of failure is  at most $e^{-t}$.
\end{proof}

We now analyze the failure probability during the iterations of phase 1.

\begin{proof}
Let us fix an iteration $i$.
We denote the discrepancy bound by $d=d_i$ where $d_i = ct^{1/2}/i^2$,  and let $\ell = m^{1-i}$ denote the number of variables.
For a row $j$ of size $s$, note that $c_j= d/\sqrt{s}$.
We call a row small if its size $s \leq s_0$, where $s_0 = d^2/(16ci^5) = ct/(16i^5)$.

The contribution of small rows to the sum $\sum_{j=1}^m \exp(-c_j^2/16)$ is at most
\[m \exp(-d^2/16s_0)  = m \exp(-ci) \leq  m \exp(-5i) \leq \ell/32\]

It remains to show that with high probability that contribution of large rows to $\sum_{j=1}^m \exp(-c_j^2/16)$ is also at most $\ell/32$. To this end, we conservatively assume that $c_j=0$ for big rows and hence we only need
to bound the probability that there are more than $\ell/32$ rows.

If we pick $\ell$ columns from the random matrix $A^{(0)}$, the expected row size is  $\mu = t \ell/m = t2^{1-i}$.
As $s_0 = ct/16i^5$ and $\mu = t 2^{1-i}$, we can pick $c$ large enough so that $ s_0 \geq 2\mu$.
By Lemma \ref{lem:prob}, the probability of having more than $\ell/32$ rows of size at least $s_0$ is bounded by
\begin{eqnarray*}
 \binom{k}{\ell} \binom{m}{\ell/32}   \Pr[B(s_0)]
& \leq &   \left(\frac{ek}{\ell} \right)^{2 \ell}  \cdot \exp(-\ell s_0 \log (s_0m/t\ell)/128) \\
& \leq & \exp(2 \ell (\log \gamma + \log m/\ell) - \ell s_0 \log (s_0 m/t\ell)/128) \\
& \leq & \exp(2 \ell (\log \gamma + i) - \ell \cdot \Omega(t/i^4)).
\end{eqnarray*}
As  $\log \gamma  = O(t^{1/2})$,  and as $i \geq \log t$ we have $\ell = m 2^{1-i} \geq m/t$, this gives an over all failure probability of $\exp(-\Omega(m/\log^4 t))$.
\end{proof} 
\section{Reducing the number of columns}
In this section we prove Theorem \ref{thm:1}.
\subsection{Fractional Discrepancy Polytope}
Let $a_1,\ldots,a_k$ be arbitrary vectors in $\R^m$.
Consider the polytope $ P:= \left\{ \sum_{i=1}^k a_i x_i : \  x_i \in [-1,1]\right\}$
of discrepancy vectors obtained by all possible fractional colorings. The convex hull of $P$ is given by its $2^k$ extreme points
\[p_\chi := \sum_i \chi_i a_i \text{ for }\chi=(\chi_1,\ldots,\chi_k) \in \{-1,+1\}^k.\]

For $p\geq 1$, let $B_p^m = \{y \in \R^m = \|y\|_p \leq 1\}$ denote the $\ell_p$ ball in $\R^m$. For brevity, let $Q := 2t B^m_\infty$.
The following lemma characterizes exactly when  $Q \subset P$.

\begin{lemma} Let $A \in \R^{m \times k}$ be the matrix with columns given by $a_1,\ldots,a_k$. Then, $Q \subset P$ iff $\|y^T A\|_1 > 2t$, for all $y \in B_1^m$.
\end{lemma}
\begin{proof}
Suppose $Q\not\subset P$. As $P$ and $Q$ are convex, by Farkas' lemma,
there exists a hyperplane given by normal $y$, that separates some point $q \in Q\setminus P$ from $P$.
As $0\in P$,  we can assume that there is some $s>0$ such that $y^T q > s$ and $y^T p_\chi < s$ for each extreme point $p_\chi$ of $P$.
As
\begin{equation}
\label{eq:l1} \max_{\chi \in \{-1,1\}^k} y^T p_\chi  =  \max_{\chi \in \{-1,1\}^k}  \sum_{i=1}^k   (y^T a_i) \chi_i = \sum_{i=1}^k |y^T a_i| = \| y^T A \|_1,
\end{equation}
this is same as saying that there is some $s$ such that  \[y^T q >  s > \| y^T A \|_1.\]
By scaling, we can assume that $\|y\|_1 =1$ and as $Q = 2t B_\infty$ we have $\max_{q \in Q}  y^Tq = 2t$. This gives that $s < 2t$, and thus there is some $y$ with $\|y\|_1=1$ and $\|y^TA\|_1 < 2t$; a contradiction.

Conversely if $Q \in P$, then no direction exists that separates some point $q\in Q$ from $P$, which implies for each $y$ with $\|y\|_1$, there is some $p_\chi \in P$ such that $y^T p_\chi > 2t$, which by \eqref{eq:l1} gives that $\|y^T A \|_1 >  2t$.
\end{proof}


So to prove Theorem \ref{thm:1}, it suffices to show that with high probability,  $\|y^TA\|_1 > 2t$ for every $y \in B_1^m$.

%
%

\subsection{Counting Argument: Weak Bound}
We start by sketching a simple but weak bound of $k = O(m^4 \log m)$.
Together with Theorem \ref{thm:2} this would already imply the $O(t^{1/2})$ discrepancy bound in Theorem \ref{thm:main}, but under the condition that $t = \Omega(\log^2 m)$.

\begin{theorem}
\label{thm:boundk}
Let $A$ be a $m\times k$ random matrix where each column has $t$ ones and $k = O(m^4 \log m)$. Then with probability at least
$1-\exp(-m \log m)$,
it holds that $\|y^T A\|_1 > 2t$ for all $y \in B_1^m$ with $\|y\|_1=1$.
\end{theorem}

Let $\delta>0$, and let $N_\delta$ be the set of points $ y' \in \R^m$ such that each coordinate $y'_i$ of $y'$ is an integral multiple of $\delta$, and $\|y'\|_1 \leq 1$.
Clearly $|N| \leq (3/\delta)^m$ and for any point $y \in B^m_1$ there is some point $y'$ in $N_\delta$ with $|y-y'|\leq m\delta$.

We fix $\delta = 3/km$ and note that $|N_\delta | \leq \exp(m  \log (km))$.
As  $\|(y-y')^T a \| \leq \|y - y'\|_1 \|a\|_\infty \leq m \delta$ for any $a \in \R^m$, to show Theorem \ref{thm:boundk} it suffices to show that
$\|y^T A\|_1 > 2t  +  km \delta$ for all $y \in N_\delta$ with $\|y\|_1 \geq 1 - m\delta \geq 9/10$.

Fix a vector $y$ in the net $N_\delta$ with $\|y\|_1 \geq 9/10$. Let $X$ denote the random variable $|y \cdot a|$, where $a \in \{0,1\}^m$ is chosen randomly with exactly  $t$ ones.
We assume henceforth that $t \leq m/10$, as for $t=\Theta(m)$, an $O(m^{1/2})$ discrepancy (even in the non-random case) follows from the result of Spencer \cite{Spencer85}.

\begin{lemma}
For every $y \in B_1^m$, $\E[X] \geq t/2m^2$ and $\E[X^2] \leq 2t/m$, assuming $t \leq m/10$.
\end{lemma}
\begin{proof}
As  $0\leq X \leq 1$, $\E[X] \geq \E[X^2]$, so it suffices to lower bound the second moment as follows.
\begin{eqnarray*} \E[X^2] & = &   \E [(\sum_i a_i y_i)^2] = \sum_i \E[a_i] y_i^2 + \sum_{i\neq j}E[a_i a_j] y_i y_j
   =  \frac{t}{m} \sum_i y_i^2  + 2 \sum_{i\neq j} \frac{t(t-1)}{m(m-1)} y_i y_j  \\
& =  & \frac{t(m-t)}{m(m-1)} \sum_i y_i^2 +  \frac{t(t-1)}{m(m-1)}  (\sum_i y_i)^2
 \geq  t/(2m) \sum_i y_i^2 \geq \frac{t}{2m^2}.
\end{eqnarray*}
Similarly, using  $(\sum_i y_i)^2 \leq 1$,
\[\E[X^2] = \frac{t}{m} \sum_i y_i^2 + \frac{t(t-1)}{m(m-1)} (\sum_i y_i)^2 \leq \frac{2t}{m}\].
\end{proof}

As $\|y^T A\|_1$ is the sum of $k$ independent random variables $X_1,\ldots,X_k$ distributed as $X$, using the lower bound on $\E[X]$ and upper bound on $\E[X^2]$ (and hence on the variance), by Lemma \ref{lem:bernstein},
\[ \Pr[ X_1+\ldots + X_k < k \frac{t}{2m^2}  -  z ] \leq \exp(\frac{- z^2/2}{(2tk/m + z/3)})  \]
Setting $k = 2z m^2/t$  and  $z=10ctm^2 \log m$, this gives
\[ \Pr[ X_1 + \ldots X_k < z ] \leq \exp(-z/10m) = \exp(-c m t \log m) \]
As $z \gg 4t$ and choosing $c$ large enough, we have
\[ |N_\delta| \exp(-c m t \log m) \ll \exp(m (\log km - tc \log m))   \ll \exp(-m \log m),\]
we obtain Theorem \ref{thm:boundk}.

\subsection{A Stronger bound}
The bound of $k = O(m^4 \log m)$ in Theorem \ref{thm:1}, combined with Theorem \ref{thm:2} implies $O(t^{1/2})$ discrepancy bound when $t = \Omega(\log^2 m)$.
So henceforth it is useful to think of $t \ll \log^2 m$. We
will now prove a refined bound of $k = O(m (\log (mt))^{O(1)})$. 

It is easy to see that $k$ must be at least $m \log m$ in general. This holds even if we only require the condition  $\|y^T A\|_1 > 2t$ to hold for $y=e_1,\ldots,e_m$. In particular, for $k\ll m \log m$, the expected number of ones in a row is $kt/m < t \log m$, and say for $t=O(1)$, it is quite likely that some row $j$ in $A$ will have fewer than $2t$ ones, and hence violate $\|e_j^T A\|_1 > 2t$.

\paragraph{The idea.}
Consider the net $N_\delta$ with $\delta = 1/km$ as before. By Theorem \ref{thm:boundk}, we can assume that $k \leq m^4 \log m$ and hence $\delta = 1/m^7$ suffices.

For a point $y \in N_{\delta}$, let $Y$ be the random variable $|y^T a|$, where $a$ is a random column with $t$ ones. We need to show that for each $y$ in the net, the sum of $k$ independent copies of the corresponding $Y$ random variables is more than $4t$. In the previous argument $k$ had to be large as the net $|N_{\delta}|$ is quite big and we were taking a union bound over all elements of $N_\delta$. While we cannot reduce the size of the net much, the idea here is to exploit the specific structure of the random vectors $a$ and the event that we care about. For instance, if $y$ is sparse, we get a not too small probability for $Y$ being small but there aren't too many sparse vectors in the net. We exploit such trade-offs below.

More precisely,  we consider another random variable $X \geq 0$ that will be stochastically dominated by $Y$,
and the value of $X$ will essentially only depend on the values of $a$ and only on the sign pattern of $y$ in certain specific coordinates (whose magnitudes are not too small).  This will lead to a much smaller loss in the union bound. We now give the details.

\subsubsection{The argument}

Fix some $y \in N_\delta$ with $\|y\|_1 \geq 1/2$.
We say that coordinate $i$ lies in class $j \in \{0,1,\ldots,h\}$, for $h=O(\log m)$, if
$|y_i| \in (2^{-j-1},2^{-j}]$.
We say that $i$ has positive sign if $y_i>0$ and negative sign if $y_i<0$.
 We will not care about coordinates that have value  $0$. Let us define the weight of class $j$ of $y$ as $w_j(y) = \sum_{i \in \text{class } j} |y_i|$.

Define the class $c(y)$ of $y$ as a class $j$ with the highest weight. Let $c^-(y)$ and $c^+(y)$ denote $c(y)-1$ and $c(y)+1$ respectively (if they exist). As $\|y\|_1 \geq 1/2$, class $c(y)$ has weight at least $1/2h$.
Let $n(y)$ denote the number of coordinates with class $c(y)$, and we thus have \[   2^{c(y)}/2h \leq   n(y)  \leq 2^{c(y)}. \]
As $c(y)$ is the maximum weight class, we also that the number of coordinates of class $c^+(y)$ and $c^-(y)$ is at most $2n(y)$ each.

We now define the random variable $X$ with the desired properties.

\paragraph{The random variable $X$.}
Let $i_1,\ldots,i_t$ denote the $t$ locations  of $1$ in $a$, that are picked from $[m]$ without replacement.
We use the principle of deferred decisions, and assume that the locations $i_1,\ldots,i_{t-1}$ have already been revealed, and that the randomness is only in the $t$-th choice.


Let $v = y_{i_1} + \ldots + y_{i_{t-1}}$, and note that
 \[Y = |y^T a| = |v + y_{i_t} |.\]
Our random variable $X$ will satisfy the following properties.
\begin{enumerate}
\item
For each $y,a$,  $X \leq Y$.
\item For every $y$, $X$ is completely determined by $v$, the sign pattern of coordinates in class $c^-(y),c(y),c^+(y)$, the location of $i_t$ in these three classes (if it falls in these classes), and on whether any of $\{i_1,\ldots,i_{t-1}\}$ fall in these three classes.
\end{enumerate}

We now define the random variable $X$, based on a few cases. The above properties are directly verified by inspection.

Let us first assume that $10 t < n(y) < m/10$. The remaining (corner) cases are handled easily later.

\paragraph{Balanced Case.} We call class $c(y)$ sign-balanced if $y$ has at least $n(y)/4$ coordinates in class $c(y)$ with
positive and negative signs.

If $v <0$, we define $X = 2^{-c(y)-1}$ if $i_t$ lies in class $c(y)$ and $y_{i_t} < 0$. Otherwise, $X=0$. Analogously, if $v>0$, then $X = 2^{-c(y)-1}$ if $i_t$ lies in class $c(y)$ and $y_{i_t}>0$. Otherwise, $X=0$.

Note that we always have $X \leq Y$. Moreover as $n(y) \geq 10 t$, irrespective of the locations of $i_1,\ldots,i_{t-1}$, the probability that $i_t$ lies in class $c(y)$ is at least $(9/10) n(y)/m$.
Finally, $X = 2^{-c(y)-1}$ with probability at least $n(y)/8m$ and at most $3n(y)/4m$, irrespective of the value of $v$.

\paragraph{Unbalanced Case.}
Without loss of generality, suppose that class $c(y)$ has more than $3n(y)/4$ positive signs (the other case is symmetric). We consider two further cases.

If $v \notin  (-2^{-c(y)+1/2},-2^{-c(y)-3/2})$, we set $X =  (1/4) 2^{-c(y)}$ if $i_t$ falls in class $c(y)$ and $y_{i_t}>0$. Otherwise, $X=0$.

Note that with the above definition $X \leq Y$. For, we either have $v \leq -2^{-c(y) + 1/2}$, in which case $v  + y_{i_t} \leq -2^{-c(y) + 1/2} + 2^{-c(y)} \leq -(1/4) 2^{-c(y)}$ so that $|v + y_{i_t}| \geq (1/4) 2^{-c(y)}$. Similarly, if $v \geq -2^{-c(y)-3/2}$, then $v + y_{i_t} \geq  -2^{-c(y)-3/2} + 2^{-c(y) - 1} \geq (1/4) 2^{-c(y)}$ so that $| v + y_{i_t}| \geq (1/4) 2^{-c(y)}$.

Finally, note that given that $i_t$ falls in class $c(y)$, $y_{i_t} > 0$ with probability at least $1/2$.

If $v \in (-2^{-c(y)+1/2},-2^{-c(y)-3/2})$, we set $X$ to $0$ if $i_t$ lies in any of the classes $\{c^-(y),c(y),c+(y)\}$.
Otherwise, we set $X = (1/16) 2^{-c(y)}$. We will still have $X \leq Y$ because if $i_t$ does not lie in any of the classes $\{c^-(y),c(y),c^+(y)\}$, then (i) either $|y_{i_t}| < 2^{-c(y) - 2}$ in which case $|v + y_{i_t}| > |v| - |y_{i_t}| > 2^{-c(y)} (1/2\sqrt{2} - 1/4) \geq (1/16) 2^{-c(y)}$; or $|y_{i_t}| > 2^{-c(y)+1}$ in which case $| v + y_{i_t}| > |y_{i_t}| - |v| > 2^{-c(y)+1} - 2^{-c(y) + 1/2} \geq 2^{-c(y) - 1}$. In either case, we have $X \leq Y$.

Moreover, as $n(y) \leq m/10$, there are  at least $m/2$ coordinates other than these three classes, so that above events happens with probability $\Omega(1)$.

\paragraph{Corner cases.} We now consider the remaining cases.
If $c(y) < 10t$,
we set $X=0$ if some $i_1,\ldots,i_{t-1}$  already
lies in $\{c^-(y),c(y),c^+(y)\}$. Otherwise, we proceed as above depending on whether $c(y)$ is balanced or  unbalanced; it is easy to check that the previous arguments still hold.

Further, the probability that any of $i_1,\ldots,i_{t-1}$ land in these three classes is at most $O(t^2/m) \ll 1$.

\vspace{2mm}

We now consider the other corner case when $c(y)$ contains more than $m/10$ coordinates. In the balanced case we proceed as previously. The problem arises in the argument above when $c(y)$ is unbalanced (since we relied on the event that $i_t$ falls outside the three classes happens with decent probability). So instead we do the following: Set $X=0$ if less than $t/20$ $i_1,\ldots,i_{t-1}$ lie inside $c(y)$. Set $X=0$ if $|v| \leq 10 \cdot 2^{-c(y)}$. Else, set $X=2^{-c(y)}$. Clearly, $X \leq Y$ as if $|v | > 10 \cdot 2^{-c(y)}$, then $|v  + y_{i_t}| > |v| - |y_{i_t}| \geq 9 \cdot 2^{-c(y)}$.

Now, the probability that fewer than $t/20$ indices $i_1,\ldots,i_{t-1},i_t$ lies inside $c(y)$ is at most $\exp(-\Omega(t))$. Further, if we condition at least $t/20$  $i_1,\ldots,i_{t-1}$ to lie in $c(y)$, then as $c(y)$ is unbalanced, the chance that $|v| \leq 10 \cdot 2^{-c(y)}$ is tiny.

\subsubsection{The concentration argument.}
Fix a $y$ and consider the random variable $X$ as defined above. Then, $\E[X] \geq (1/16) 2^{-c(y)} \cdot n(y)/m$ which is at least $c/mh$ for a fixed constant $c > 0$. Moreover, $X$ is bounded by $2^{-c(y)} \leq 1/n(y)$. Therefore, by Bernstein's bound, if we choose $k$ independent copies of $X_1,\ldots,X_k$ of $X$, then for $\mu = \E[X] \geq ck/2mh$,
$$\Pr[ X_1 + \ldots + X_k < \mu/2] \leq \exp\left(\frac{- \mu^2/8}{(2\mu/n(y)) + (\mu/6n(y))}\right)= \exp(-\Omega(n(y) \mu)) = \exp\left(-\Omega\left( \frac{n(y) k}{m \log m}\right)\right).$$

We now use our definition of the random variable $X$ and a union bound over what the random variable $X$ can depend on. There are about $2/\delta$ choices of $v$. Now, consider some $y$ of class $c(y)$. The behavior of $X$ is completely determined by the sign pattern on $O(n(y))$ coordinates of $y$ (corresponding to the sign-pattern of $y$ restricted to classes $\{c^-(y), c(y), c^+(y)\}$). So in the union bound, we incur an additional loss of $\exp(5n(y))$ (as $|c^-(y)|, |c^+(y)| \leq 2 n(y)$). Further, we have at most $\binom{m}{n(y)} \cdot \binom{m}{2n(y)} \cdot \binom{m}{2n(y)}$ possibilities for $\{c^-(y), c(y), c^+(y)\}$. Therefore, taking a union bound over all possible random variables $X$, we get the failure probability for a fixed $n(y)$ to be at most
$$\exp(-\Omega( n(y) k/m \log m)) \cdot (em)^{5 n(y)} \ll \exp(-\Omega(n(y) C \log m)),$$
if we take $k = C m \log^2 m$ for a sufficiently big constant $C$. Adding over all values of $n(y)$ we get that the failure probability in Theorem \ref{thm:1} is at most $m^{-\Omega(C)}$ for $k = C m \log^2 m$ for $C$ sufficiently big. This finishes the proof of Theorem 2. 

\section*{Acknowledgements}
The authors would like to thank the Simons Institute at Berkeley for hosting the two authors when the present was done.

\bibliographystyle{plain}
\bibliography{refr}

\end{document}